\documentclass[11pt,a4paper]{article}
\usepackage{authblk}
\usepackage[latin1]{inputenc}
\usepackage[english]{babel}
\usepackage{amsmath,hyperref,amsfonts,amssymb}
\usepackage{enumerate}
\usepackage{latexsym}
\usepackage{rotating}
\usepackage{dsfont}
\usepackage{fancyhdr}
\usepackage{indentfirst}
\usepackage{graphicx}
\usepackage{newlfont}
\usepackage{latexsym}
\usepackage{amsthm}
\usepackage{array}
\usepackage[numbers]{natbib}
\PassOptionsToPackage{hyphens}{url}\usepackage{hyperref}
\usepackage{hyphenat}

\newtheorem{theorem}{\bf Theorem}[section]

\newtheorem{proposition}[theorem]{\bf Proposition}

\newtheorem{corollary}[theorem]{\bf Corollary}
\newtheorem{definition}[theorem]{\bf Definition}
\newtheorem{remark}[theorem]{\bf Remark}

\newcommand{\beq}{\begin{equation}}
\newcommand{\bea}[1]{\begin{array}{#1} }
\newcommand{\eeq}{ \end{equation}}
\newcommand{\ea}{ \end{array}}

\newcommand{\ud}{\mathrm{d}}
 
\newcommand{\incrx}[1]{\left(#1(t_{i+1}^n)-#1(t_i^n)\right)}

\newcommand{\fqv}[1] {finite quadratic variation along $#1$ }
\newcommand{\prop}[1] {Proposition \ref{prop:#1}}
\newcommand{\thm}[1] {Theorem \ref{thm:#1}}

\newcommand{\defin}[1] {Definition \ref{def:#1}}

\newcommand{\rmk}[1] {Remark \ref{rmk:#1}}

\newcommand{\Sec}[1] {Section \ref{sec:#1}}

\newcommand{\eq}[1] {\eqref{eq:#1}}

\newcommand{\norm}[1]{\left\lVert #1 \right\rVert} 
\newcommand{\abs}[1]{\left\lvert #1 \right\rvert} 

\def \t {{\tau}}

\def \d {{\delta}}
\def \k {{\kappa}}

\def \s {{\sigma}}
\def \w {{\omega}}
\def \e {{\epsilon}}

\def\l {\lambda}

\def \O {{\Omega}}
\def \De {{\Delta}}

\def \L {{\Lambda}}

\def \R  {{\mathbb {R}}} 
\def \N {{\mathcal {N}}}
\def \C {{\mathcal {C}}}

\def \P {{\cal{P}}}
\def \F {\mathcal{F}}  
\def \H {\mathcal{H}}
\def \K {\mathcal{K}}
\def \S {\mathcal{S}}
\def \Si {\Sigma}

\def \W {\mathcal{W}}  

\def \NN {{\mathbb {N}}}
\def \VV  {{\mathbb {V}}}

\def \PP  {{\mathbb {P}}}

\def \FF  {{\mathbb {F}}}  
\def \CC {\mathbb{C}}
\def \BB {\mathbb{B}}
\def \Ft {\left(\mathcal{F}_t\right)_{t\in[0,T]}}  
\def \Cloc {\mathbb{C}^{1,2}_{loc}}
\def \Cb {\mathbb{C}^{1,2}_{b}}

\def \DT {D([0,T],\R^d)}
\def \Ft {\left(\mathcal{F}_t\right)_{t\in[0,T]}}  
\def \Cloc {\mathbb{C}^{1,2}_{loc}}
\def \Cb {\mathbb{C}^{1,2}_{b}}

\def \DT {D([0,T],\R^d)}

\def \ind {{\mathds{1}}}

\def \Limn {\lim_{n\rightarrow \infty}}
\def \limn {\xrightarrow[n\rightarrow \infty]{}}

\def \zs {, \qquad t\in[0,T]}

\def \ito {It\^o}
\def \follmer {F\"{o}llmer}
\def \cadlag {c\`adl\`ag}
\def \caglad {c\`agl\`ad}

\def \lf {\left(}
\def \rg {\right)}
\def \ps {$(\O,\F,(\F_t)_{0\leq t\leq T},\PP)$}
\def \OT {$[0,T]$}

\def \naf {non-anticipative functional}
\def \vd {\nabla_{\w}}

\def \hd {\mathcal{D}}
\def \dinf {\textrm{d}_{\infty}}

\def \tr {\mathrm{tr}}

\title{A pathwise approach to continuous-time trading
\footnote{We gratefully acknowledge financial support from the ERC (249415-RMAC), the NATIXIS Foundation for Quantitative Research, Scuola Normale Superiore di Pisa, and the Isaac Newton Institute for Mathematical Sciences (Cambridge, UK).}}
\author{Candia Riga\thanks{University of Zurich, Email: candia.riga@uzh.ch}}
\date{January 28, 2016}

\begin{document}
\maketitle

\begin{abstract}
This paper develops a mathematical framework for the analysis of continuous\hyp time trading strategies which, in contrast to the classical setting of continuous-time mathematical finance, does not rely on stochastic integrals or other probabilistic notions.

Our purely analytic framework allows for the derivation of a pathwise self-financial condition for continuous-time trading strategies, which is consistent with the classical definition in case a probability model is introduced.
Our first proposition provides us with a pathwise definition of the gain process for a large class of continuous-time, path-dependent, self-finacing trading strategies,  including the important class of \lq delta-hedging\rq\ strategies, and is based on the recently developed \lq non-anticipative functional calculus\rq. Two versions of the statement involve respectively continuous and \cadlag\ price paths.
The second proposition is a pathwise replication result that generalizes the ones obtained in the classical framework of diffusion models. Moreover, it gives an explicit and purely pathwise formula for the hedging error of delta-hedging strategies for path-dependent derivatives across a given set of scenarios.
We also provide an economic justification of our main assumption on price paths.



\end{abstract}

\section{Introduction}

Since the emergence of modern mathematical finance, the common framework has been to model the financial market as a filtered probability space \ps\ on which the prices of liquid traded stocks are represented by stochastic processes $X=(X_t)_{t\geq0}$ and the payoffs of derivatives as functionals of the underlying price process. The probability measure $\PP$, also called \textit{real world, historical, physical} or \textit{objective} probability, tries to capture the observed patterns and, in the equilibrium interpretation, represents the (subjective) expectation of the ``representative investor''. 
However, the choice of an objective probability measure is not obvious and always encompasses a certain amount of model risk and model ambiguity.
Recently, there has been a growing emphasis on the dangerous consequences of relying on a specific probabilistic model.
The concept of the so-called \textit{Knightian uncertainty}, introduced way back in 1921 by Frank Knight~\citep{knight} while distinguishing between ``risk'' and ``uncertainty'', is still as relevant today and led to a new challenging research area in Mathematical Finance.
More fundamentally, the existence of a single objective probability does not even make sense, agreeing with the criticism raised by \citet{definetti31,definetti37}.

In the growing flow of literature addressing the issue of model ambiguity, we may recognize two approaches: \textbf{model-independent}, where the single probability measure $\PP$ is replaced by a family $\P$ of plausible probability measures, and \textbf{model-free}, that eliminates probabilistic a priori assumptions altogether, and relies instead on pathwise statements.
This paper takes a model-free approach, which also provides a solution to another problem affecting the classical probabilistic modeling of financial markets.
Indeed, in continuous-time financial models, the gain process of a self-financing trading strategy is represented as a stochastic integral.
However, despite the elegance of the probabilistic representation, some real concerns arise.
Beside the issue of the impossible consensus on a probability measure, the representation of the gain from trading lacks a pathwise meaning: while being a limit in probability of approximating Riemann sums, the stochastic integral does not have a well-defined value on a given \lq state of the world\rq.
This causes a gap in the use of probabilistic models, in the sense that it is not possible to compute the gain of a trading portfolio given the realized trajectory of the underlying stock price, which constitutes a drawback in terms of financial interpretation.

Beginning in the nineties, a new branch of the literature has addressed the issue of pathwise integration in the context of financial mathematics.
A breakthrough in this direction was the seminal paper written by \citet{follmer} in 1981. He proved a pathwise version of the \ito\ formula, conceiving the construction of an integral of a $C^1$-class function of a \cadlag\ path with respect to that path itself, as a limit of non-anticipative Riemann sums. His purely analytical approach does not ask for any probabilistic structure, which may instead come into play only at a later point by considering stochastic processes that satisfy almost surely, i.e. for almost all paths, the analytical requirements. In this case, the so-called \emph{F\"ollmer integral} provides a path-by-path construction of the stochastic integral.
F\"ollmer's framework turns out to be of main interest in finance (see also \cite{schied-CPPI}, \cite[Sections 4,5]{follmer-schied}, and \cite[Chapter 2]{sondermann}) as it allows to avoid any probabilistic assumption on the dynamics of traded stocks and consequently to avoid any model risk/ambiguity. Reasonably, only observed price trajectories are involved.
In 1994, \citet{bickwill} provided an interesting economic interpretation of F\"ollmer's pathwise calculus, leading to new perspectives in the mathematical modeling of financial markets.
Bick and Willinger reduced the computation of the initial cost of a replicating trading strategy to an exercise of analysis. Moreover, for a given price trajectory (state of the world), they showed one is able to compute the outcome of a given trading strategy, that is the gain from trade. 
Other contributions towards the pathwise characterization of stochastic integrals have been obtained via probabilistic techniques by Wong and Zakai (1965), \citet{bichteler}, \citet{karandikar} and \citet{nutz-int} (only existence), and via convergence of discrete-time economies by \citet{willtaq}.

In this paper, we set our framework in a similar way to \cite{bickwill}, and we enhance it  by the aid of the pathwise calculus for \naf s, introduced by \citet{dupire} and developed by \citet{contf2010}. This theory extends the F\"ollmer's pathwise calculus to a large class of non-anticipative functionals.

Once we have at our disposal a pathwise notion of the gain from trading, it would be interesting to examine the analytical conditions imposed on the price paths in relationship with no-arbitrage-like notions within a model-free framework.

An important series of papers on this subject comes from Vladimir Vovk (see e.g. \citet{vovk-vol,vovk-proba,vovk-rough,vovk-cadlag}). 
 He introduced an outer measure (see \cite[Definition 1.7.1]{tao} for the definition of \emph{outer measure}) on the space of possible price paths, called \emph{upper price} (\defin{upperP} below), as the minimum super-replication price of a very special class of European contingent claims. The important intuition behind this notion of upper price is that the sets of price paths with zero upper price, called \emph{null sets}, allow for the infinite gain of a positive portfolio capital with unitary initial endowment. The need to guarantee some type of market efficiency in a financial market leads to discard the null sets.
\citeauthor{vovk-proba} says that a property holds for \emph{typical paths} if the set of paths where it does not hold is null, i.e. has zero upper price. 
Let us give some details.
\begin{definition}[Vovk's upper price]\label{def:upperP}
  The \emph{upper price} of a set $E\subset\O$ is defined as 
\beq\label{eq:upperP}
\bar\PP(E):=\inf_{S\in\S}\{S(0)|\,\forall\w\in\O,\; S(T,\w)\geq\ind_E(\w)\},
\eeq
where $\S$ is the set of all \emph{positive capital processes} $S$, that is: $S=\sum_{n=1}^\infty\K^{c_n,G_n}$, where $\K^{c_n,G_n}$ are the portfolio values of bounded simple predictable strategies trading at a non-decreasing infinite sequence of stopping times $\{\t^n_i\}_{i\geq1}$, such that for all $\w\in\O$ $\t^n_i(\w)=\infty$ for all but finitely many $i\in\NN$, with initial capitals $c_n$ and with the constraints $\K^{c_n,G_n}\geq0$ on $[0,T]\times\O$ for all $n\in\NN$ and $\sum_{n=1}^\infty c_n<\infty$.
\end{definition}
It is immediate to see that $\bar\PP(E)=0$ if and only if there exists a positive capital process $S$ with initial capital $S(0)=1$ and infinite capital at time $T$ on all paths in $E$, i.e. $S(T,\w)=\infty$ for all $\w\in E$.

Depending on what space $\O$ is considered, Vovk obtained specific results. In particular, he investigated properties of typical paths that concern their measure of variability. The most general framework considered is $\O=D([0,T],\R_+)$. He proved in \cite{vovk-rough} that typical paths $\w$ have a \emph{$p$-variation index} less or equal to 2, which means that the $p$-variation is finite for all $p>2$, but we have no information for $p=2$ (a stronger result is stated in \cite[Proposition 1]{vovk-rough}). If we relax the positivity and we restrict to \cadlag\ path with all components having \lq moderate jumps\rq\ in the sense of \eq{mod-jumps}, then \citet{vovk-cadlag} obtained appealing results regarding the quadratic variation of typical paths along special sequences of random partitions.
Indeed, by adding a control on the size of the jumps, in the sense of considering the sample space $\O_\psi$, defined by
\beq\label{eq:mod-jumps}
\O_\psi:=\left\{\w\in D([0,T],\R)\bigg|\,\forall t\in(0,T],\;\abs{\De\w(t)}\leq\psi\lf \sup_{s\in[0,t)}\abs{\w(s)}\rg\right\}
\eeq
for a given non-decreasing function $\psi:[0,\infty)$, \citet{vovk-cadlag} obtained finer results. 
In particular, he proved that typical paths have a property that we call \emph{Vovk's quadratic variation} (\defin{qv-vovk} in the Appendix) along a special nested sequence of random time partitions. 
The same result holds in the multi-dimensional case, that is on the space $$\O_{\psi}^d:=\left\{\w=(\w_1,\ldots,\w_d)\in D([0,T],\R^d):\  \w_i\in\O_{\psi},\ i=1,\ldots,d\right\}.$$
Moreover, it applies to all nested sequences of random partitions of dyadic type (\cite[Proposition 1]{vovk-cadlag}), and any two sequences of dyadic type give the same value of Vovk's quadratic variation for typical paths (\cite[Proposition 2]{vovk-cadlag}). A nested sequence of partitions is called of \emph{dyadic type} for the coordinate process on $D([0,T],\R)$ if it is composed of stopping times such that there exist a polynomial $p$ and a constant $C>0$ and
\begin{enumerate}
\item for all $\w\in\O_\psi$, $n\in\NN_0$, $0\leq s<t\leq T$, if $\abs{\w(t)-\w(s)}>C2^{-n}$, then there is an element of the $n^{th}$ partition which belongs to $(s,t]$,
\item for typical $\w$, from some $n$ on, the number of finite elements of the $n^{th}$ partition is at most $p(n)2^{2n}$.
\end{enumerate}

When the sample space is $C([0,T],\R)$ \cite{vovk-proba} proved that typical paths are either constant or have a $p$-variation which is finite for all $p>2$ and infinite for $p\leq2$ (stronger results are stated in \cite[Corollaries 4.6,4.7]{vovk-proba}.
Note that the situation changes remarkably from the space of \cadlag\ paths to the space of continuous paths. Indeed, no (positive) \cadlag\ path which is bounded away from zero and has finite total variation can belong to a null set in $D([0,T],\R^d_+)$, while all continuous paths with finite total variation belong to a null set in $C([0,T],\R^d_+)$.

A similar notion of outer measure is introduced by \citet{perk-promel} (see also \citet{perkowski-thesis}), which is more intuitive in terms of hedging strategies. He considers portfolio values that are limits of simple predictable portfolios with the same positive initial capital and whose correspondent simple trading strategies never risk more than the initial capital.
\begin{definition}[Definition 3.2.1 in \cite{perkowski-thesis}]\label{def:outerP}
  The \emph{outer content} of a set $E\subset\O:=C([0,T],\R^d)$ is defined as 
\beq\label{eq:outerP}
\tilde\PP(E):=\inf_{(H^n)_{n\geq1}\in\H_{\l,s}}\{\l|\,\forall\w\in\O,\; \liminf_{n\to\infty}(\l+(H^n\bullet\w)(T))\geq\ind_E(\w)\},
\eeq
where $\H_{\l,s}$ is the set of all \emph{$\l$-admissible simple strategies}, that is of bounded simple predictable strategies $H^n$ trading at a non-decreasing infinite sequence of stopping times $\{\t^n_i\}_{i\geq1}$, $\t^n_i(\w)=\infty$ for all but finitely many $i\in\NN$ for all $\w\in\O$, such that $(H^n\bullet\w)(t)\geq-\l$ for all $(t,\w)\in[0,T]\times\O$.
\end{definition}
Analogously to Vovk's upper price, the $\tilde\PP$-null sets are identified with the sets where the inferior limit of some sequence of 1-admissible simple strategies brings infinite capital at time $T$. This characterization is shown to be a model-free interpretation of the condition of \emph{no arbitrage of the first kind} (NA1) from mathematical finance, also referred to as \emph{no unbounded profit with bounded risk} (see e.g. \cite{kk2007,kardaras}). Indeed, in a financial model where the price process is a semimartingale on some probability space $(\O,\F,\PP)$, the (NA1) property holds if the set $\{1+(H\bullet S)(T),\,H\in\H_1\}$ is bounded in $\PP$-probability, i.e. if 
$$\lim_{c\to\infty}\sup_{H\in\H_{1,s}}\PP(1+(H\bullet S)(T)\geq c)=0.$$
On the other hand, \cite[Proposition 3.28]{perkowski-thesis} proved that an event $A\in\F$ which is $\tilde\PP$-null has zero probability for any probability measure on $(\O,\F)$ such that the coordinate process satisfies (NA1).
However, the characterization of null sets in \cite{perk-promel,perkowski-thesis} is possibly weaker than Vovk's one. In fact, the outer measure $\tilde\PP$ is dominated by the outer measure $\bar\PP$.


A distinct approach to a model-free characterization of arbitrage is proposed by \citet{riedel}, although he only allows for static hedging. He considers a Polish space $(\O,\ud)$ with the Borel sigma-field and he assumes that there are $D$ uncertain assets in the market with known non-negative prices $f_d\geq0$ at time 0 and uncertain values $S_d$ at time $T$, which are continuous on  $(\O,\ud)$, $d=1,\ldots,D$. A portfolio is a vector $\pi$ in $\R^{D+1}$ and it is called an \emph{arbitrage} if $\pi\cdot f\leq0$, $\pi\cdot S\geq0$ and $\pi\cdot S(\w)>0$ for some $\w\in\O$, where $f_0=S_0=1$. Thus the classical ``almost surely'' is replaced by ``for all scenarios'' and ``with positive probability'' is replaced by ``for some scenarios''.
The main theorem in \cite{riedel} is a model-free version of the FTAP and states that the market is \emph{arbitrage-free} if and only if there exists a \emph{full support martingale measure}, that is a probability measure whose topological support in the polish space of reference is the full space and under which the expectation of the final prices $S$ is equal to the initial prices $f$.
This is proven thanks to the continuity assumption of $S(\w)$ in $\w$ on one side and a separation argument on the other side. Even without a prior probability assumption, it shows that, if there are no (static) arbitrages in the market, it is possible to introduce a pricing probability measure, which assigns positive probability to all open sets.

The sequence of trading times partitions considered in the present paper is both dense and nested, condition under which our notion of quadratic variation is equivalent to Vovk's one (see \cite[Section 1.1.1]{riga-thesis} for an investigation of different notions of pathwise quadratic variation and their relationship).

The outline of the paper is the following. 
Section \ref{sec:functionals} introduces the notation and reviews some key concepts and results of the pathwise functional calculus introduced in \cite{contf2010}.

In \Sec{setting}, we set our analytical framework and we start by defining \emph{simple trading strategies}, whose trading times are covered by the elements of a given sequence $\Pi$ of partitions of the time horizon $[0,T]$ and for which the self-financing condition is straightforward.
In \Sec{self-fin}, we provide equivalent self-financing conditions for (non-simple) trading strategies on any set of paths, whose gain from trading is the limit of gains of simple strategies and satisfies the pathwise counterpart equation of the classical self-financing condition. Similar conditions were assumed in \cite{bickwill} for convergence of general trading strategies.
\Sec{gain} contains the first main result: in \prop{G} for the continuous case and in \prop{G-cadlag} for the \cadlag\ case, we obtain the path-by-path computability of the gain of path-dependent trading strategies in a certain class of self-financing strategies on the set of paths with finite quadratic variation along $\Pi$.
For example, in the continuous case, for dynamic risky stock positions $\phi$ in the vector space of \emph{vertical 1-forms}, 
the gain of the corresponding self-financing trading strategy is well-defined as a \cadlag\ process $G(\cdot,\cdot;\phi)$ such that
 \begin{align*}
   G(t,\w;\phi)=&\int_0^t \phi(u,\w_u)\cdot\ud^{\Pi}\w \\
   =&\lim_{n\rightarrow\infty}\sum_{t^n_i\in\pi^n, t^n_i\leq t}\phi(t_i^n,\w^{n}_{t^n_i})\cdot(\w(t_{i+1}^n)-\w(t_i^n))
  \end{align*}
for all continuous paths of finite quadratic variation along $\Pi$, where $\w^n$ is a piecewise constant approximation of $\w$ defined in \eq{wn}. 
In \Sec{replication}, we present a pathwise replication result, \prop{hedge}, that can be seen as the model-free and path-dependent counterpart of the well known pricing PDE in mathematical finance, giving furthermore an explicit formula for the \emph{hedging error}. That is, if a \lq smooth\rq\ \naf\ $F$ solves
$$\left\{\bea{ll}
\hd F(t,\w_t)+\frac12\tr\lf A(t)\cdot\vd^2F(t,\w_t)\rg=0,\quad t\in[0,T), \w\in Q_A(\Pi) \\
F(T,\w)=H(\w),
\ea\right.$$
were $H$ is a continuous (in sup norm) payoff functional and $Q_A(\Pi)$ is the set of paths with absolutely continuous quadratic variation along $\Pi$ with density $A$, then the profit and loss of the self-financing trading strategy on $Q(\O^0,\Pi)$ with initial investment $F(0,\cdot)$ and stock holdings process $\vd F$ against the sale of the (path-dependent) contingent claim $H$ with maturity $T$ is
\beq\label{eq:er}
\frac12\int_{0}^T\tr\lf\lf A(t)-\tilde A(t)\rg\vd^2F(t,\w_t)\rg \ud t
\eeq
in all price scenarios $\w\in Q_{\tilde A}(\Pi)$.
In particular, the delta-hedging strategy described above has a value that is equal to $H$ at time $T$ in any scenario in $Q_A(\Pi)$ and to $F(t,\w_t)$ at any time $t\in[0,T]$ in any scenario $\w\in Q(\O^0,\Pi)$.
The explicit error formula \eq{er} is the purely analytical counterpart of the probabilistic formula  given in \cite{elkaroui}, where a mis-specification of volatility is considered in a stochastic framework, furthermore extended to the path-dependent case.
Finally, in \Sec{plausible}, we discuss on how the condition of finite quadratic variation that we assume on price paths may be justified from an economic point of view. 
The Appendix contains a comparison between our notion of quadratic variation and the other pathwise notions in the literature.

\section{Notation: pathwise functional calculus}\label{sec:functionals}

We first introduce some notations and summarize some key concepts and results of the pathwise functional calculus introduced in \cite{contf2010,cont-notes} that we use throughout this paper. We use the usual notation $\DT$ for the set of $\R^d$-valued \cadlag\ functions on \OT, that we generically call \textit{\cadlag\ paths}.
Let $t\in[0,T]$. Denote, for a \cadlag\ path $x\in\DT$,
\begin{itemize}
\item $x(t)\in\R^d$ its value at $t$;
\item $x_t=x(t\wedge\cdot)\in\DT$ its path \lq stopped\rq\ at time $t$;
\item $x_{t-}=x\ind_{[0,t)}+x(t-)\ind_{[t,T]}\in\DT$;
\item for $\d\in\R^d$, $x_t^\d=x_t+\d\ind_{[t,T]}\in\DT$ the \textit{vertical perturbation} of  $x$ in direction $\d$ of the  'future' portion of the path.
\end{itemize}

Then, define the \textit{space of stopped paths} $$\L_T:=\{(t,x_t):\:(t,x)\in[0,T]\times\DT\}$$ as  the quotient of $[0,T]\times\DT$ by the equivalence relation $$\forall(t,x),(t',x')\in[0,T]\times\DT,\quad(t,x)\sim(t',x') \iff t=t',x_t=x'_{t}.$$ The space $\L_T$ is equipped with a distance $\dinf$, defined by
$$\dinf((t,x),(t',x'))=\sup_{u\in[0,T]}|x(u\wedge t)-x'(u\wedge t')|+|t-t'|=||x_t-x'_{t'}||_\infty+|t-t'|,$$
for all $(t,x),(t',x')\in\L_T$.
$(\L_T,\dinf)$ is then a complete metric space and the subset of continuous stopped paths,
$$\W_T:=\{(t,x)\in\L_T:\,x\in C([0,T],\R^d)\}$$
is a closed subspace of $(\L_T,\dinf)$.
A \textit{non-anticipative functional} on $\DT$ is defined as a measurable map $F:(\L_T,d_\infty)\to\R^d$.
\begin{definition}\label{def:regF}A \naf\ $F$ is:
\begin{itemize}
\item \emph{jointly-continuous} if $F:(\L_T,d_\infty)\to\R^d$ is a continuous map.
\item \emph{continuous at fixed times} if for all $t\in[0,T],$\\
$$F(t,\cdot):( \DT ,||\cdot||_\infty )\mapsto\R $$
is continuous.
\item \emph{left-continuous}, i.e. $F\in\CC_l^{0,0}(\L_T)$, if
$$\bea{c}\forall (t,x)\in\L_T, \forall\e>0,\,\exists\eta>0:\quad \forall h\in[0,t],\,\forall (t-h,x')\in\L_T,\\
\quad \dinf((t,x),(t-h,x'))<\eta\quad\Rightarrow\quad|F(t, x)-F(t-h,x')|<\e;\ea$$
One can similarly define the set $\CC_r^{0,0}(\L_T)$ of \emph{right-continuous} functionals.
\item \emph{boundedness-preserving}, i.e. $F\in\BB(\L_T)$, if,
$$\bea{c}\forall K\subset\R^d\text{ compact, }\forall t_0\in[0,T],\,\exists C_{K,t_0}>0;\quad \forall t\in[0,t_0],\,\forall (t,x)\in\L_T,\\
x([0,t])\subset K \Rightarrow \quad \forall t\in[0,t_0],\quad |F(t,x)|<C_{K,t_0}.\ea$$
\end{itemize}
\end{definition}
\begin{remark}\label{rmk:regularity}
Note that, from the regularity of \naf s, useful pathwise regularities follow \cite{contf2010}:
\begin{enumerate}
\item If $F\in\CC_l^{0,0}(\L_T)$ then, for all $x\in\DT$, the path  
$$[0,T]\to\R^d,\quad t\mapsto F(t,x_{t-})$$ is left-continuous;
\item If $F\in\CC_r^{0,0}(\L_T)$ then, for all $x\in\DT$, the path 
$$[0,T]\to\R^d,\quad t\mapsto F(t,x_t)$$ is right-continuous;
\item If $F\in\CC^{0,0}(\L_T)$ then, for all $x\in\DT$, the path 
$$[0,T]\to\R^d,\quad t\mapsto F(t,x_t)$$ is \cadlag\ and continuous at each point where $x$ is continuous.
\end{enumerate}
\end{remark}
We recall now the notions of differentiability for \naf s.
\begin{definition}\label{def:derF}
  A \naf\ $F$ is called
\begin{itemize}
\item \emph{horizontally differentiable at} $(t,\omega)\in\L_T$ if the limit
$${\cal D} F(t,\omega)=\lim_{h\rightarrow0^+}\frac{F(t+h,\omega_{t})-F(t,\omega_t)}{h}$$
exists and is finite, in which case it is denoted $\hd F(t,x)$ and called the horizontal derivative of $F$ at $(t,\w)$.
If this holds for all $(t,\omega)\in\L_T,  t<T$,  the \naf\ $\hd F:(t,\omega)\mapsto {\cal D} F(t,\omega)$ is called the \emph{horizontal derivative} of $F$.
\item \emph{vertically differentiable at} $(t,\omega)\in\L_T$ if the map
$$e\in\R^d \mapsto F(t, \omega_t+e 1_{[t,T]})$$ is differentiable at 0. Its gradient at 0 is then denoted 
$\vd F(t, \omega)$ and called the vertical derivative of $F$ at $(t,\omega).$
If this holds for all $(t,x)\in\L_T$, then the $\R^d$-valued \naf\ $\vd F$ is called the \emph{vertical derivative} of $F$.
\end{itemize}
\end{definition}
We denote
\begin{itemize}
\item $\CC^{1,2}(\L_T)$ the set of \naf s $F$ which are
\begin{itemize}
\item horizontally differentiable with $\hd F$ continuous at fixed times,
\item two times vertically differentiable with $\vd^j F\in\CC^{0,0}_l(\L_T)$ for $j=1,2$;
\end{itemize}
\item $\CC^{1,2}_b(\L_T)$ the set of \naf s $F\in\CC^{1,2}(\L_T)$ such that $\hd F,\vd F,\ldots,\vd^2F\in\BB(\L_T)$.
\end{itemize}
The following weaker regularity condition still allows key results to hold. A \naf\ $F$ is said to be \emph{locally regular}, i.e. $F\in\Cloc(\L_T)$, if $F\in\CC^{0,0}(\L_T)$ and there exist a sequence of stopping times $(\t_k)_{k\geq1}$, $\t_0=0$, $\t_k\to_{k\to\infty}\infty$, and a family of \naf s $\{F^k\in\Cb(\L_T),\}_{k\geq0},$
such that
$$F(t,x_t)=\sum_{k\geq 0}F^k(t,x_t)\ind_{[\t_k( x),\t_{k+1}(x))}(t) \zs.$$
The interest in these classes of non-anticipative functionals stems from the functional change of variable formula \cite{contf2010,cont-notes} shown in \thm{fif-d} below.

Let $\Pi=\{\pi_n\}_{n\geq1}$ be a sequence of partitions of $[0,T]$, that is for all $n\geq1$ $\pi_n=(t_i^n)_{i=0,\ldots,m(n)},\;0=t_0^n<\ldots<t_{m(n)}^n=T$. We say that $\Pi$ is \emph{dense} if $\cup_{n\geq1}\pi_n$ is dense in $[0,T]$, or equivalently the mesh $\abs{\pi^n}:=\max_{i=1,\ldots m(n)}|t^n_i-t^n_{i-1}|$ goes to 0 as $n$ goes to infinity, and we say that $\Pi$ is \emph{nested} if $\pi_{n+1}\subset\pi_{n}$ for all $n\in\NN$. 
\begin{definition}\label{def:qv1}
  Let $\Pi$ be a dense sequence of partitions of $[0,T]$, a \cadlag\ function $x:[0,T]\to\R$ is said to be of \emph{finite quadratic variation along $\Pi$} if there exists a non-negative \cadlag\ function $[x]_\Pi:[0,T]\to\R_+$ such that
\beq\label{eq:qv}
\forall t\in[0,T],\quad[x]_\Pi(t)=\Limn\sum_{\stackrel{i=0,\ldots,m(n)-1:}{t^n_i\leq t}}(x(t^n_{i+1})-x(t^n_{i}))^2<\infty
\eeq
and 
\begin{equation} \label{eq:qv-jumps}
[x]_\Pi(t)=[x]_\Pi^c(t)+\sum_{0<s\leq t}\De x^2(s) \zs,
\end{equation}
where $[x]_\Pi^c$ is a continuous non-decreasing function and $\De x(t):=x(t)-x(t-)$ as usual.
In this case, the non-decreasing function $[x]_\Pi$ is called the \emph{quadratic variation of $x$ along $\Pi$}.
\end{definition}
Note that the quadratic variation $[x]_\Pi$ strongly depends on the sequence of partitions $\Pi$. Indeed, as remarked in \cite[Example 2.18]{cont-notes}, for any real-valued continuous function we can construct a sequence of partition along which that function has null quadratic variation.

In the multi-dimensional case, the definition is modified as follows.
\begin{definition}\label{def:qvd}
A \cadlag\ path $x:[0,T]\to\R^d$, $t\mapsto\,^t(x^1(t),\ldots,x^d(t))$, is of \emph{finite quadratic variation along $\Pi$} if, for all $1\leq i,j\leq d$, $x^i,x^i+x^j$ are of finite quadratic variation along $\Pi$. In this case, the function $[x]_\Pi$ has values in the set $\S^+(d)$ of positive symmetric $d\times d$ matrices:
$$[x]_\Pi(t)=\Limn\sum_{\stackrel{i=0,\ldots,m(n)-1:}{t^n_i\leq t}}\incrx{x}\,^t\!\incrx{x},\quad  t\in[0,T],$$
whose elements are given by
\begin{eqnarray*}
([x]_\Pi)_{i,j}(t) &=& \frac12\lf[x^i+x^j]_\Pi(t)-[x^i]_\Pi(t)-[x^j]_\Pi(t)\rg \\
 &=& [x^i,x^j]_\Pi^c(t)+\sum_{0<s\leq t}\De x^i(s)\De x^j(s) , \quad\qquad i,j=1,\ldots d.
\end{eqnarray*}
\end{definition}

For any set $U$ of \cadlag\ paths with values in $\R$ (or $\R^d$), we denote by $Q(U,\Pi)$ the subset of $U$ of paths having \fqv{\Pi}.

Note that $Q(D([0,T],\R),\Pi)$ is not a vector space, because assuming $x^1,x^2\in\penalty0 Q(D([0,T],\R),\Pi)$ does not imply $x^1+x^2\in Q(D([0,T],\R),\Pi)$ in general. This is the reason of the additional requirement $x^i+x^j\in Q(D([0,T],\R),\Pi)$ in \defin{qvd}. As remarked in \cite[Remark 2.20]{cont-notes}, the subset of paths $x$ being $C^1$-functions of a same path $\w\in D([0,T],\R^d)$, i.e. 
$$\{x\in Q(D([0,T],\R),\Pi),\;\exists f\in C^1(\R^d,\R),\,x(t)=f(\w(t))\,\forall t\in[0,T]\},$$
is instead closed with respect to the quadratic variation composed with the sum of two elements.

Henceforth, we will fix a sequence of partitions $\Pi$ and, when considering a path $x\in Q(U,\Pi)$, we will drop the subscript in the notation of its quadratic variation, denoting $[x]$ instead of $[x]_\Pi$.

Now we can state the main result of the pathwise functional calculus.
\begin{theorem}[Functional change of variable formula \cite{contf2010}]\label{thm:fif-d}
$\ $\\Let $x\in Q(\DT,\pi)$ such that
\begin{equation}  \label{eq:ass_w}
\sup\limits_{t\in[0,T]-\pi^n}|\De x(t)|\limn0
\end{equation}
and  denote:
\beq\label{eq:wn}
x^n:=\sum_{i=0}^{m(n)-1}x(t^n_{i+1}-)\ind_{[t^n_i,t^n_{i+1})}+x(T)\ind_{\{T\}},\eeq
a piecewise constant approximation of $x$ along $\pi$. 
Then, for $F\in\Cloc(\L_T)$
the limit
$$\Limn\sum_{i=0}^{m(n)-1}\vd F \lf t_i^n,x^{n,\De x(t^n_i)}_{t^n_i-}\rg(x(t_{i+1}^n)-x(t_i^n))$$
exists and, denoted by 
\beq\label{eq:int} \int_{0}^T\vd F(t,x_{t-})\cdot\ud^{\pi}x(t),\eeq
it satisfies
\begin{align}\label{eq:fif-d}  
 F(T,x_T)={}& F(0,x_{0})+\int_{0}^T\vd F(t,x_{t-})\cdot\ud^{\pi}x(t)\\
&{}+\int_{0}^T\hd F(t,x_{t-)}\ud t + \int_{0}^T\frac12\tr\lf\vd^2F(t,x_{t-})\cdot \ud[x]^c(t)\rg \nonumber\\
&{}+\sum_{t\in(0,T]}\lf F(t,x_t)-F(t,x_{t-})-\vd F(t,x_{t-})\cdot\De x(t)\rg. \nonumber
\end{align}
\end{theorem}
Following \cite{contf2010}, we refer to \eqref{eq:int} as the \textit{F\"ollmer integral} of $\nabla_{\omega} F$ with respect to $x$ along $\pi$. It is defined for any $x\in Q(\DT, \pi)$.
Note that the assumption \eq{ass_w} can always be removed by  including all jump times of the \cadlag\ path $x$ in the  sequence of partitions $\pi$.

\section{The setting}
\label{sec:setting}

We consider a continuous-time frictionless market open for trade during the time interval $[0,T]$, where $d$ risky (non-dividend-paying) stocks, named `stock', as well as a riskless security, named \lq bond', are traded. The latter is assumed to be the numeraire security and we refer directly to the forward stock and portfolio values, which makes this framework of simplified notation without loss of generality. 
Our setting does not make use of any (subjective) probabilistic assumption on the market dynamics and we construct trading strategies based on the realized paths of the stock prices. 

Precisely, we consider the metric space $(\O,||\cdot||_\infty)$, composed of the set of $\R^d$-valued non-negative \cadlag\ paths, $\O:=D([0,T],\R^d_+)$, equipped with the sup norm. Then, we equip it with the Borel sigma-field $\F$ and the canonical filtration $\FF=\Ft$, that is the natural filtration of the coordinate process $S$:
$$\bea{ll}S(t,\w)=\w(t),& \w\in\O,\ t\in[0,T],\\\F_t=\s\lf\{S(u),\,u\in[0,t]\}\rg,& t\in[0,T].\ea$$
Within the financial market, $\O$ represents the space of all possible trajectories of the stock prices up to the time horizon $T$, also called \textit{scenarios} or \textit{price paths}. When considering only continuous price paths, we will restrict to the subspace $\O^0:=C([0,T],\R^d_+)$.


In such analytical framework, we think of a continuous-time path\hyp dependent trading strategy as determined by the value of the initial investment and the quantities of stock and bond holdings at any time in \OT, the latter being functions of both time and price path.
\begin{definition}
  A \emph{trading strategy} in $(\O,\F)$ is any triple $(V_0,\phi,\psi)$, where $V_0:\O\to\R$ is $\F_0$-measurable and $\phi=(\phi(t,\cdot))_{t\in(0,T]},\psi=(\psi(t,\cdot))_{t\in(0,T]}$ are $\FF$-adapted \caglad\ processes on $(\O,\F)$, respectively with values in $\R^d$ and in $\R$. The portfolio value $V$ of such trading strategy at any time $t\in[0,T]$ and in any scenario $\w\in\O$ is given by $$V(t,\w;\phi,\psi)=\phi(t,\w)\cdot\w(t)+\psi(t,\w).$$
\end{definition}
Economically speaking,  the elements of the vector $\phi(t,\w)$ represent the number of each stock to be held in the trading portfolio at time $t$ in the scenario $w\in\O$, and $\psi(t,\w)$ represents the units of money invested in the bond at time $t$ in the scenario $w\in\O$. The left-continuity of the holding processes comes from the fact that any revision of the portfolio will be executed the instant just after the time the decision is made. On the other hand, their right-continuous modifications $\phi(t+,\w),\psi(t+,\w)$, defined by 
$$\phi(t+,\w):=\lim\limits_{s\searrow t}\phi(s,\w),\ \psi(t+,\w):=\lim\limits_{s\searrow t}\psi(s,\w),\quad\forall\w\in\O,\,t\in[0,T)$$
represent respectively the number of stocks and bonds in the portfolio just after any revision of the trading portfolio decided at time $t$.
The choice of strategies adapted to the canonical filtration conveys the realistic assumption that any trading decision makes use only of the price information available at the time it takes place.

We aim to identify \emph{self-financing trading strategies} in this pathwise framework, that is portfolios where changes in the stock position are necessarily financed by buying or selling bonds without adding or withdrawing any cash. In particular, we will look for those of them which trade continuously in time but still allow for an explicit computation of the gain from trading.
Unlike discrete-time models,  in the classical literature about continuous-time financial market models we don't have a general pathwise characterization of self-financing dynamic trading strategies, mainly because the self-financing condition ivolves the gain process, which is defined in terms of a stochastic integral with respect to the stock price process and is thus a purely probabilistic object. In the same way, the number of bonds which continuously rebalances the portfolio does not have a determined value in a given scenario.

\section{Self-financing strategies}
\label{sec:self-fin}

We start by considering strategies where the portfolio is rebalanced only a finite number of times, for which the self-financing condition is well established and whose gain is given by a Riemann sum.

Henceforth, we will take as given a dense nested sequence of time partitions, $\Pi=(\pi^n)_{n\geq1}$, i.e. $\pi^n=\{0=t^n_0<t^n_1<\ldots,t^n_{m(n)}=T\}$, $\pi^n\subset\pi^{n+1}$, $\abs{\pi^n}\limn\infty$.

We denote by $\Si(\Pi)$ the set of simple predictable processes whose jump times are covered by one of the partitions in $\Pi$\footnote{We could assume in more generality that the jump times are only covered by $\cup_{n\geq1}\pi^n$, but at the expense of more complicated formulas}:
\begin{align*}
\Si(\pi^n):={}&\bigg\{\phi:\;\forall i=0,\ldots,m(n)-1,\;\exists \l_i\,\F_{t^n_i}\mbox{-measurable }\R^d\mbox{-valued}\\
&\quad\mbox{random variable on }(\O,\F),\;\phi(t,\w)=\sum_{i=0}^{m(n)-1}\l_i(\w)\ind_{(t^n_i,t^n_{i+1}]}\bigg\},\\
\Si(\Pi):={}&\underset{n\geq1}\cup\Si(\pi^n).
\end{align*}

\begin{definition}
  $(V_0,\phi,\psi)$ is called a \emph{simple self-financing trading strategy} if it is a trading strategy such that $\phi\in\Si(\pi^n)$ for some $n\in\NN$,
  \beq\label{eq:simplephi}\phi(t,\w)=\sum_{i=0}^{m(n)-1}\l_i(\w)\ind_{(t^n_i,t^n_{i+1}]},\quad \l_i:(\O,\F_{t^n_i})\to\R^d\text{ measurable }\forall i\eeq 
  and
\begin{align}
\psi(t,\w;\phi)={}&V_0-\phi(0+,\w)\cdot\w(0)\label{eq:psi-sf}\\
&{}-\sum_{i=1}^{m(n)-1}\w(t^n_i\wedge t)\cdot(\phi(t^n_{i+1}\wedge t,\w)-\phi(t^n_i\wedge t,\w)) \nonumber\\
  ={}& V_0-\phi(0+,\w)\cdot\w(0)-\sum_{i=1}^{k(t,n)}\w(t^n_i)\cdot(\l_{i}(\w)-\l_{i-1}(\w)),\nonumber
\end{align}
where $k(t,n):=\max\{i\in\{1,\ldots,m\}\;:\;t^n_i<t\}$.
The \emph{gain} of such a strategy is defined at any time $t\in[0,T]$ by
\begin{align*}
G(t,\w;\phi):={}&\sum_{i=1}^{m(n)}\phi(t^n_{i}\wedge t,\w)\cdot(\w(t^n_{i}\wedge t)-\w(t^n_{i-1}\wedge t)) \\
  ={}& \sum_{i=1}^{k(t,n)}\l_{i-1}(\w)\cdot(\w(t^n_{i})-\w(t^n_{i-1}))+\l_{k(t,n)}(\w)\cdot(\w(t)-\w(t^n_{k(t,n)})).
\end{align*}
\end{definition}
In the following, when there is no ambiguity, we drop the dependence of $k$ on $t,n$ and write $k\equiv k(t,n)$.

Note that the condition \eq{psi-sf} is equivalent to requiring that the trading strategy $(V_0,\phi,\psi)$ satisfies
$$V(t,\w;\phi,\psi)=V_0+G(t,\w;\phi).$$

Since a simple self-financing trading strategy is uniquely determined by its initial investment and the stock position at all times, we will drop the dependence on $\psi$ of the quantities involved. For instance, when we refer to a simple self-financing strategy $(V_0,\phi)$, we implicitly refer to the triplet $(V_0,\phi,\psi)$ with $\psi\equiv\psi(\cdot,\cdot;\phi)$ defined in \eq{psi-sf}, and we denote by $V(t,\w;\phi)\equiv V(t,\w;\phi,\psi)$ its portfolio value.

\begin{remark}
The portfolio value $V(\cdot,\cdot;\phi)$ of a simple self-financing strategy $(V_0,\phi)$ is a real-valued $\FF$-adapted \cadlag\ process on $(\O,\F)$, satisfying
$$\Delta V(t,\w;\phi)=\phi(t,\w)\cdot\Delta\w(t),\quad \forall t\in[0,T],\w\in\O.$$
\end{remark}
The right-continuity of $V$ comes from the definition \eq{psi-sf}, which implies, for all $t\in[0,T]$ and $\w\in\O$,
$$\psi(t,\w)+\phi(t,\w)\cdot\w(t)=\psi(t+,\w)+\phi(t+,\w)\cdot\w(t).$$

Below, we establish the self-financing conditions for (non\hyp simple) trading strategies.

\begin{definition}\label{def:path-sf} 
  Given an $\F_0$-measurable random variable $V_0:\O\to\R$ and an $\R^d$-valued $\FF$-adapted \caglad\ process $\phi=(\phi(t,\cdot))_{t\in(0,T]}$ on $(\O,\F)$, we say that $(V_0,\phi)$ is a \emph{self-financing trading strategy on} $U\subset\O$ if there exists a sequence of self-financing simple trading strategies $\{(V_0,\phi^n,\psi^n), n\in\NN\}$, such that 
$$\forall\w\in U,\,\forall t\in[0,T],\quad\phi^n(t,\w)\limn\phi(t,\w),$$
and any of the following conditions is satisfied:
 \begin{enumerate}[(i)]
 \item there exists a real-valued $\FF$-adapted \cadlag\ process $G(\cdot,\cdot;\phi)$ on $(\O,\F)$ such that, for all $t\in[0,T],\w\in U$,
$$G(t,\w;\phi^n)\limn G(t,\w;\phi)\quad\text{and}\quad\Delta G(t,\w;\phi)=\phi(t,\w)\cdot\Delta\w(t);$$
 \item there exists a real-valued $\FF$-adapted \cadlag\ process $\psi(\cdot,\cdot;\phi)$ on $(\O,\F)$ such that, for all $t\in[0,T],\w\in U$,
$$\psi^n(t,\w;\phi^n)\limn\psi(t,\w;\phi)$$
and $$\psi(t+,\w;\phi)-\psi(t,\w;\phi)=-\w(t)\cdot\lf\phi(t+,\w)-\phi(t,\w)\rg;$$
 \item there exists a real-valued $\FF$-adapted \cadlag\ process $V(\cdot,\cdot;\phi)$ on $(\O,\F)$ such that, for all $t\in[0,T],\w\in U$,
$$V(t,\w;\phi^n)\limn V(t,\w;\phi)\quad\text{and}\quad\Delta V(t,\w;\phi)=\phi(t,\w)\cdot\Delta\w(t).$$
 \end{enumerate}
\end{definition}

\begin{remark}\label{rmk:path-sf}
It is easy to see that the three conditions (i)-(iii) of Definition \ref{def:path-sf} are equivalent. If any of them is fulfilled, the limiting processes $G,\psi,V$ define respectively the gain, bond holdings and portfolio value of the self-financing strategy $(V_0,\phi)$ on $U$ and they satisfy, for all $t\in[0,T],\w\in U$,
\begin{equation}
   \label{eq:sf}
V(t,\w;\phi)=V_0+G(t,\w;\phi)
 \end{equation}
and
\begin{equation}
  \label{eq:psi}
  \psi(t,\w;\phi)=V_0-\phi(0+,\w)-\Limn\sum_{i=1}^{m(n)}\w(t^n_i\wedge t)\cdot(\phi^n(t^n_{i+1}\wedge t,\w)-\phi^n(t^n_{i}\wedge t,\w)).
\end{equation}
\end{remark}
Equation \eq{sf} is the pathwise counterpart of the classical definition of self-financing in probabilistic financial market models. However, in our purely analytical framework, we couldn't take it directly as a self-financing condition because some prior assumptions are needed to define path-by-path the quantities involved.

\section{Pathwise construction of the gain process}
\label{sec:gain}

In the following two propositions we show that we can identify a special class of (pathwise) self-financing trading strategies, respectively on the set of continuous price paths with \fqv{\Pi} and on the set of \cadlag\ price paths with \fqv{\Pi}, whose gain is computable path-by-path as a limit of Riemann sums.

We define the following space of stock holdings strategies:
\begin{align}
\VV:=\big\{&\phi\ \R^d\text{-valued }\FF\text{-adapted \cadlag\ process}:\ \exists F\in\Cloc(\W_T)\cap\CC^{0,0}(\W_T), \nonumber\\
&\phi(t,\w)=\vd F(t,\w_{t-})\quad\forall\w\in Q(\O,\Pi),t\in[0,T]\big\}. \label{eq:nabla}
\end{align}
Note that $\VV$ has a natural structure of vector space; we call its elements \emph{vertical 1-forms}.

\begin{proposition}[Continuous price paths]\label{prop:G}
 Let $\phi=(\phi(t,\cdot))_{t\in(0,T]}\in\VV$, i.e.
  \begin{equation}
    \phi(t,\w)=\vd F(t,\w_{t-})\quad \forall\w\in Q(\O,\Pi),t\in[0,T],
 \end{equation}
where $F\in\Cloc(\W_T)\cap\CC^{0,0}(\W_T)$. Then, there exists a \cadlag\ process $G(\cdot,\cdot;\phi)$ such that, for all $\w\in Q(\O^0,\Pi)$ and $t\in[0,T]$,
 \begin{align}
   G(t,\w;\phi)={}&\int_0^t \phi(u,\w_u)\cdot\ud^{\Pi}\w \label{eq:fi}\\
   \label{eq:pathint}
   ={}&\lim_{n\rightarrow\infty}\sum_{t^n_i\leq t}\vd F(t_i^n,\w^{n}_{t^n_i-})\cdot(\w(t_{i+1}^n\wedge T)-\w(t_i^n\wedge T)),
  \end{align}
where $\w^n$ is defined as in \eq{wn}.
Moreover, $\phi$ is the stock holdings process of a self-financing trading strategy on $Q(\O^0,\Pi)$ with gain process $G(\cdot,\cdot;\phi)$.
\end{proposition}
\begin{proof}
  First of all, under the assumptions, the change of variable formula for functionals of continuous paths holds (\citep[Theorem 3]{contf2010}), which ensures the existence of the limit in \eq{pathint} and provide us with the definition of the F\"ollmer integral in \eq{fi}.
Then, we observe that each $n^{th}$ sum in the right-hand side of \eqref{eq:pathint} is exactly the accumulated gain of a pathwise self-financing strategy which trades only a finite number of times. Precisely, let us define, for all $\w\in\O$ and all $t\in[0,T)$, 
$$ \phi^n(t,\w):=\phi(0+,\w)\ind_{\{0\}}(t)+\sum_{i=0}^{m(n)-1}\phi\left(t^n_{i},\w^{n}\right)\ind_{(t^n_{i},t^n_{i+1}]}(t),$$
and $\psi^n\equiv\psi^n(\cdot,\cdot;\phi^n)$ defined according to \eq{psi-sf},
then $(\phi^n,\psi^n)$ are the sotck and bond holdings processes of a simple self-financing strategy, with cumulative gain $G(\cdot,\cdot;\phi^n)$ given by
\begin{align*} 
G(t,\w;\phi^n)={}&\sum_{i=1}^{k}\vd F\lf t^n_{i-1},\w^{n}_{t^n_{i-1}-}\rg\cdot(\w(t^n_{i})-\w(t^n_{i-1}))\\
&\,+\vd F\lf t^n_{k},\w^{n}_{t^n_{k}-}\rg\cdot(\w(t)-\w(t^n_{k})).
\end{align*}
and portfolio value $V(\cdot,\cdot;\phi^n)$ given by
$$V(t,\w;\phi^n)=\psi^n(t,\w)+\phi^n(t,\w)\cdot\w(t)=V_0+G(t,\w;\phi^n).$$
Then, we have to check whether the simple approximation $(\phi^n,\psi^n)$ satisfies the self-financing conditions for a trading strategy with stock holdings process $\phi$. 
First, $\phi^n$ converges pointwise to $\phi$ on $[0,T]\times Q(\O^0,\Pi)$, i.e.
$$|\phi^n(t,\w)-\phi(t,\w)|\limn0\quad \forall\w\in Q(\O^0,\Pi),\, t\in[0,T].$$
Indeed, for any $t\in[0,T],\w\in Q(\O^0,\Pi)$ and $\e>0$, there exist $\bar n\in\NN$ and $\eta>0$, such that, for all $n\geq\bar n$, 
\begin{align*}
\dinf\left((t^n_k,\w^{n}_{t_k^n-}),(t,\w)\right)={}&\max\left\{||\w^n_{t^n_k-},\w_{t^n_k-}||_\infty,\sup_{u\in[t^n_k,t)}|\w(t^n_k)-\w(u)|\right\}\\
&{}+|t-t_k^n|\\
<{}&\eta,
\end{align*}
where $k\equiv k(t,n):=\max\{i\in\{1,\ldots,m\}\;:\;t^n_i<t\}$, and then
\begin{align*}
|\phi^n(t,\w)-\phi(t,\w)|={}&\abs{\phi(t_k^n,\w_{t_k^n}^{n})-\phi(t,\w)}\\
={}&\abs{\vd F(t_k^n,\w_{t_k^n-}^{n})-\vd F(t,\w)}\\
\leq{}&\e,
\end{align*}
because by assumption $\vd F\in\CC_l^{0,0}(\W_T)$. 
We have thus built a sequence of self-financing simple trading strategies approximating $\phi$. 
Then, the convergence of the related gains on the set of continuous paths with finite quadratic variation along $\Pi$ is given by the definition of the \follmer\ integral in \citep[Theorem 3]{contf2010}: for all $t\in[0,T]$ and $\w\in Q(\O^0,\Pi)$,
$$G(t,\w;\phi^n)\limn G(t,\w;\phi),\quad G(t,\w;\phi)=\int_0^t\vd F(u,\w)\cdot\ud^{\Pi}\w.$$
Moreover, by the assumptions on $F$ and by \rmk{regularity}, the map $t\mapsto F(t,\w_{t})$ is continuous for all $\w\in C([0,T],\R^d)$. Therefore, by the change of variable formula for functionals of continuous paths, $G(\cdot,\w;\phi)$ is continuous for all $\w\in Q(\O^0,\Pi)$.
Thus, the process $G(\cdot,\cdot;\phi)$ satisfies the condition (i) in Definition \ref{def:path-sf} and so it is the well-defined gain process of any self-financing trading strategy with stock holdings process $\phi$, on $ Q(\O^0,\Pi)$.
\end{proof}

\begin{corollary}\label{cor:path-sf}
  Let $\phi$ be as in \prop{G} and $V_0:\O\to\R$ be any $\F_0$-measurable function, then the \caglad\ process $\psi(\cdot,\cdot;\phi)$, defined for all $t\in[0,T]$ and $\w\in Q(\O^0,\Pi)$ by
  \begin{align*}
  \psi(t,\w;\phi)={}&V_0-\phi(0+,\w)\\
&{}-\Limn\sum_{i=1}^{k(t,n)}\w(t^n_i)\cdot\lf\vd F\lf t^n_{i},\w^{n}_{t^n_{i}-}\rg-\vd F\lf t^n_{i-1},\w^{n}_{t^n_{i-1}-}\rg\rg,
  \end{align*}
is the bond holdings process of the self-financing trading strategy $(V_0,\phi)$ on $Q(\O^0,\Pi)$.
\end{corollary}

In order to get a version of \prop{G} on \cadlag\ paths, we impose slightly stricter conditions.

\begin{proposition}[C\`adl\`ag price paths]\label{prop:G-cadlag}
 Let  $\phi=(\phi(t,\cdot))_{t\in(0,T]}$ be an $\R^d$-valued $\FF$-adapted \caglad\ process on $(\O,\F)$ and assume that there exists a smooth \naf
 $$F\in\Cloc(\L_T)\cap\CC^{0,0}_r(\L_T),\qquad\vd F\in\CC^{0,0}(\L_T),$$ satisfying
  \begin{equation*}
    \phi(t,\w)=\vd F(t,\w_{t-})\quad\forall\w\in Q(\O,\Pi),\,t\in[0,T].
  \end{equation*}
Then, there exists a \cadlag\ process $G(\cdot,\cdot;\phi)$ such that, for all $\w\in Q(\O,\Pi)$ and $t\in[0,T]$,
 \begin{align}
   G(t,\w;\phi)={}&\int_0^t \phi(u,\w_u)\cdot\ud^{\Pi}\w  \label{eq:pathint-cadlag}\\
   ={}&\lim_{n\rightarrow\infty}\sum_{t^n_i\leq t}\vd F\lf t_i^n,\w^{n,\De\w(t_i^n)}_{t^n_i-}\rg\cdot(\w(t_{i+1}^n\wedge T)-\w(t_i^n\wedge T)),    \nonumber
  \end{align}
where $\w^n$ is defined as in \eq{wn}.
Moreover, $\phi$ is the stock position of a self-financing trading strategy on $Q(\O,\Pi)$ with gain process $G(\cdot,\cdot;\phi)$.
\end{proposition}
\begin{proof}
The essence of the proof follows the lines of the proof of \prop{G}, using the change of variable formula for functionals of \cadlag\ paths, instead of continuous paths, which entails the definition of the \follmer\ integral \eq{pathint-cadlag}.
For all $\w\in\O$ and $t\in[0,T]$, we define
$$  \phi^n(t,\w):=\phi(0,\w)\ind_{\{0\}}(t)+\sum_{i=0}^{m(n)-1}\phi\left(t^n_{i}+,\w^{n,\De\w(t^n_{i})}_{t^n_i-}\right)\ind_{(t^n_{i},t^n_{i+1}]}(t) $$
and $\psi^n\equiv\psi^n(\cdot,\cdot;\phi^n)$ defined according to \eq{psi-sf},
then $(\phi^n,\psi^n)$ are the sotck and bond holdings processes of a simple self-financing strategy, with cumulative gain $G(\cdot,\cdot;\phi^n)$ given by
\begin{align*}
G^n(t,\w)={}&\sum_{i=1}^{k}\vd F\lf t^n_{i-1},\w^{n,\De\w(t^n_{i-1})}_{t^n_{i-1}-}\rg\cdot(\w(t^n_{i})-\w(t^n_{i-1}))\\
&\,+\vd F\lf t^n_{k},\w^{n,\De\w(t^n_{k})}_{t^n_{k}-}\rg\cdot(\w(t)-\w(t^n_{k})), 
\end{align*}
Then, we verify that
$$\forall\w\in Q(\O,\Pi),\,\forall t\in[0,T],\quad|\phi^n(t,\w)-\phi(t,\w)|\limn0.$$
This is true, by the left-continuity of $\vd F$: for each $t\in[0,T],\w\in Q(\O,\Pi)$ and $n\in\N$, we have that
$\forall \e>0$, $\exists \eta=\eta(\e)>0$, $\exists\bar n=\bar n(t,\eta)\in\NN$ such that, $\forall n\geq\bar n$,
\begin{align*}
\dinf\left(\w^{n,\De\w(t^n_k)}_{t_k^n-},\w_{t-}\right)={}&\max\left\{||\w^n_{t^n_k-},\w_{t^n_k-}||_\infty,\sup_{u\in[t^n_k,t)}|\w(t^n_k)-\w(u)|\right\}\\
&{}+|t-t_k^n|\\
<{}&\eta,
\end{align*}
hence
\begin{align*}
|\phi^n(t,\w)-\phi(t,\w)|={}&\abs{\lim_{s\searrow t^n_k}\phi(s,\w_{t_k^n-}^{n,\De\w(t^n_k)})-\phi(t,\w)}\\
={}&\lim_{s\searrow t^n_k}\abs{\vd F(s,\w_{t_k^n-}^{n,\De\w(t^n_k)})-\vd F(t,\w_{t-})}\\
\leq&\e.
\end{align*}
Therefore:
$$G(t,\w;\phi^n)\limn G(t,\w;\phi),\quad G(t,\w;\phi)=\int_{(0,t]}\vd F(u,\w_{u-})\cdot\ud^{\Pi}\w,$$
where $G(t,\w;\phi)$ is a  real-valued $\FF$-adapted process on $(\O,\F)$.
Moreover, by the change of variable formula \eq{fif-d} and \rmk{regularity}, it is \cadlag\ with left-side jumps given by
\begin{align*}
  \De G(t,\w;\phi)={}&\lim_{s\nearrow t}(G(t,\w;\phi)-G(s,\w;\phi))\\
  ={}& F(t,\w_{t})- F(t,\w_{t-})\\
  &{}-\lf F(t,\w_{t})- F(t,\w_{t-})-\vd F(t,\w_{t-})\cdot\De\w(t)\rg\\
  ={}&\vd F(t,\w_{t-})\cdot\De\w(t).
\end{align*}
Therefore,  the condition (i) in Definition \ref{def:path-sf} is satisfied.
\end{proof}

\begin{corollary}
  Let $\phi$ be as in \prop{G-cadlag} and $V_0:\O\to\R$ be any $\F_0$-measurable function, then the \caglad\ process $\psi(\cdot,\cdot;\phi)$, defined for all $t\in[0,T]$ and $\w\in Q(\O,\Pi)$ by
\begin{align*}
\psi(t,\w;\phi)={}&V_0-\phi(0+,\w)\\
&{}-\Limn\sum_{i=1}^{k(t,n)}\w(t^n_i)\cdot\lf\vd F_{t^n_{i}}\lf\w^{n,\De\w(t^n_{i})}_{t^n_{i}-}\rg-\vd F_{t^n_{i-1}}\lf\w^{n,\De\w(t^n_{i-1})}_{t^n_{i-1}-}\rg\rg
\end{align*}
is the bond holdings process of the self-financing trading strategy $(V_0,\phi)$ on $Q(\O,\Pi)$. 
\end{corollary}

\section{Pathwise replication of contingent claims}
\label{sec:replication}

A non-probabilistic replication result restricted to the non-path-dependent case was obtained by \citet[Proposition 3]{bickwill}, even if lacking a formula for the hedging error in case of non replication.
Here, we state the generalization to the replication problem for path-dependent contingent claims, furthermore providing an explicit  formula for the hedging error on certain classes of price paths.


\begin{definition}\label{def:hedging_error}
  The \emph{hedging error} of a self-financing trading strategy $(V_0,\phi)$ on $U\subset D([0,T],\R^d_+)$ for a path-dependent $T$-contingent claim $H$ in a scenario $\w\in U$ is the value 
$$V(T,\w;\phi)-H(\w)=V_0(\w)+G(T,\w;\phi)-H(\w).$$
$(V_0,\phi)$ is said to \emph{replicate $H$ on $U$}  if its hedging error for $H$ is null on $U$, while it is called a \emph{super-strategy for $H$ on $U$} if its hedging error for $H$ is non-negative on $U$, i.e.
$$V_0(\w)+ G(T,\w;\phi)\geq H(\w_T)\quad\forall\w\in U.$$
\end{definition}

For any \cadlag\ function with values in $\S^+(d)$, say $A\in D([0,T],\S^+(d))$, we denote by
$$Q_A(\Pi):=\left\{\w\in Q(\O,\Pi):\;[\w](t)=\int_0^tA(s)\ud s\quad\forall t\in[0,T]\right\}$$
the set of price paths of finite quadratic variation along $\Pi$, whose quadratic variation is absolutely continuous with density $A$. Note that the elements of $Q_A(\Pi)$ are continuous, by \eq{qv-jumps}.
\begin{proposition}\label{prop:hedge}
Consider a path-dependent contingent claim with exercise date $T$ and a continuous payoff functional $H:(\O,\norm{\cdot}_\infty)\mapsto\R$. Assume that there exists a smooth \naf\ $F\in\Cloc(\W_T)\cap\CC^{0,0}(\W_T)$ that satisfies
\beq\label{eq:fpde1}
\left\{\bea{ll}
\hd F(t,\w_t)+\frac12\tr\lf A(t)\vd^2F(t,\w_t)\rg=0,& t\in[0,T),\w\in Q_A(\Pi)\\
F(T,\w)=H(\w).&
\ea\right.\eeq
Let  $\tilde A\in D([0,T],\S^+(d))$. Then, the hedging error of the trading strategy $(F(0,\cdot),\vd F)$, self-financing on $Q(\O^0,\Pi)$, for $H$ in any price scenario $\w\in Q_{\tilde A}(\Pi)$ is 
\begin{equation}\label{eq:err}
  \frac12\int_{0}^T\tr\lf (A(t)-\tilde A(t))\vd^2F(t,\w_t)\rg \ud t.
\end{equation}
In particular, the trading strategy $(F(0,\cdot),\vd F)$ replicates the contingent claim $H$ on $Q_A(\Pi)$, and its portfolio value at any time $t\in[0,T]$ in any scenario $\w\in Q_A(\Pi)$ is given by $F(t,w_t)$.
\end{proposition}

\begin{proof}[Proof]
By \prop{G}, the gain at time $t\in[0,T]$ of the trading strategy $(F(0,\cdot),\vd F)$ in a price scenario $\w\in Q(\O^0,\Pi)$ is given by 
$$G(t,\w;\vd F)=\int_0^t\vd F(u,\w_u)\cdot\ud^{\Pi}\w(u).$$
Moreover, this strategy is self-financing on $Q(\O^0,\Pi)$, hence, by Remark \ref{rmk:path-sf}, its portfolio value at any time $t\in[0,T]$ in any scenario $\w\in Q(\O^0,\Pi)$ is given by
$$V(t,\w)=F(0,\w_{0})+\int_0^t\vd F(u,\w_u)\cdot\ud^{\Pi}\w(u).$$ 
In particular, since $F$ is smooth, we can apply the change of variable formula for functionals of continuous paths. By using the functional partial differential equation \eqref{eq:fpde1}, for all $\w\in Q_{\tilde A}(\Pi)$, this gives 
\begin{align*}
V(T,\w)={}&F(0,\w_{0})+\int_{0}^T\vd F(t,\w_t)\cdot\ud^{\Pi}\w(t) \\
 ={}&F(T,\w_T)-\int_{0}^T\hd F(t,\w_t)\ud t-\frac12\int_{0}^T\tr\lf \tilde A(t)\vd^2F(t,\w_t)\rg \ud t\\
={}&H-\frac12\int_{0}^T\tr\lf(\tilde A(t)-A(t))\vd^2F(t,\w_t)\rg \ud t.
\end{align*}
\end{proof}

\section{A plausibility requirement}
\label{sec:plausible}

We want now to discuss the reasonability of the ever-present assumption of finte quadratic variation on price paths.

The results on typical price paths reviewed in the Introduction cannot be directly applied to our framework, because we want to work with a fixed sequence of time partitions rather than with a random one.
Nonetheless, we can deduce that if we consider a singleton $\{\w\}$, where $\w\in\O_\psi$, with $\O_\psi$ defined in \eq{mod-jumps}, and our sequence $\Pi$ of partitions is of dyadic type for $\w$, then the property of finite quadratic variation for $\w$ is necessary to prevent the existence of a positive capital process (\defin{upperP}) trading at times in $\Pi$, that starts from a finite initial capital but ends up with infinite capital at time $T$.
However, the conditions imposed on the sequence of partitions are difficult to check.

Instead, we turn around the perspective: we keep our deterministic sequence $\Pi$ of partitions fixed and try to identify the right subset of paths in $\O$ that is \emph{plausible} working with. 
To do so, we propose the following notion of \emph{plausibility} that, together with a technical condition on the paths, suggests that it is indeed a good choice to work on set of price paths with finite quadratic variation along $\Pi$.
\begin{definition}
  A set of paths $U\subset\O$ is called \emph{plausible} if there does not exist a sequence $(V_0^n,\phi^n)$ of simple self-financing strategies such that:
  \begin{enumerate}[(i)]
  \item the correspondent sequence of portfolio values, $\{V(t,\w;\phi^n)\}_{n\geq1}$, is non-decreasing for all paths $\w\in U$ at any time $t\in[0,T]$,
  \item the correspondent sequence of initial investments $\{V^n_0(\w_0)\}_{n\geq1}$ converges for all paths $\w\in U$,
  \item the correspondent sequence of gains along some path $\w\in U$ at the final time $T$ grows to infinity with $n$, i.e. $G(T,\w;\phi^n)\limn\infty$.
  \end{enumerate}
\end{definition}

\begin{proposition}
Let $U\subset\O$ be a set of price paths satisfying, for all $(t,\w)\in[0,T]\times U$ and all $n\in\NN$, 
\beq\label{eq:cn-conv}
\sum_{n=1}^\infty\!\lf\!\sum_{i=0}^{m(n-1)-1}\!\!\!\!\!\!\sum_{\stackrel{j,k:\,j\neq k,}{t^{n-1}_i\leq t_j^n,t^n_k<t^{n-1}_{i+1}}}\!\!\!\!\!\!(\w(t^n_{j+1}\wedge t)-\w(t^n_j\wedge t))\cdot(\w(t^n_{k+1}\wedge t)-\w(t^n_k\wedge t))\!\rg^-
\eeq
is finite, where $(x)^-:=\max\{0,-x\}$ denotes the negative part of $x\in\R$. 
Then, if $U$ is plausible, all paths $\w\in U$ have \fqv{\Pi}.
\end{proposition}
Before proceding with the proof, let us translate the condition \eq{cn-conv} in terms of the relation between the price path $\w$ and the sequence of nested partitions $\Pi$.
Let $d=1$ for sake of notation. Denote by $A^n$ the $n^{th}$-approximation of the quadratic variation along $\Pi$, i.e.
$$A^n(t,\w):=\sum_{i=0}^{m(n)-1}(\w(t^n_{i+1}\wedge t)-\w(t^n_i\wedge t))^2\quad\forall(t,\w)\in[0,T]\times\O.$$
Then:
\begin{align*}
  &A^n(t,\w)-A^{n-1}(t,\w)=\\
={}&\sum_{i=0}^{m(n)-1}(\w(t^n_{i+1}\wedge t)-\w(t^n_i\wedge t))^2-\sum_{i=0}^{m(n-1)-1}(\w(t^{n-1}_{i+1}\wedge t)-\w(t^{n-1}_i\wedge t))^2\\
={}&\sum_{i=0}^{m(n-1)-1}\!\lf\sum_{t^{n-1}_i\leq t_j^n< t^{n-1}_{i+1}}\!\!(\w(t^n_{j+1}\wedge t)-\w(t^n_j\wedge t))^2-(\w(t^{n-1}_{i+1}\wedge t)-\w(t^{n-1}_i\wedge t))^2\rg\\
={}&{}-2\sum_{i=0}^{m(n-1)-1}\!\sum_{\stackrel{j,k:\,j\neq k,}{t^{n-1}_i\leq t_j^n,t^n_k<t^{n-1}_{i+1}}}\!(\w(t^n_{j+1}\wedge t)-\w(t^n_j\wedge t))(\w(t^n_{k+1}\wedge t)-\w(t^n_k\wedge t)).
\end{align*}
Thus the series in \eq{cn-conv} is exactly the series $\sum_{n=1}^\infty (A^n(t,\w)-A^{n-1}(t,\w))^-$.

\proof
For $n\in\NN$, let us define a simple predictable process $\phi^n\in\Si(\pi^n)$ by
  \begin{align}
    \label{eq:Vn}
    \phi^n(t,\w):={}&{}-2\sum_{i=0}^{m(n)-1}\w(t^n_i)\ind_{(t^n_i,t^n_{i+1}]}(t)
  \end{align}
Then, we can rewrite the $n^{\mathrm{th}}$ approximation of the quadratic variation of $\w$ at time $t\in[0,T]$ as 
\begin{align}
  A^n(t,\w)={}&\w(t)^2-\w(0)^2-2\sum_{i=0}^{m(n)-1}\w(t^n_i)(\w(t^n_{i+1}\wedge t)-\w(t^n_i\wedge t))\nonumber\\
={}&\w(t)^2-\w(0)^2+G(t,\w;\phi^n)\nonumber\\
={}&V(t,\w;\phi^n)-c_n,  \label{eq:An}
\end{align}
where $c_n=\w(0)^2-\w(t)^2+V^n_0(\w_0)$.
We want to define the initial capitals $V^n_0$ in such a way that the sequence of simple self-financing strategies $(V_0^n,\phi^n)$ has non decreasing portfolio values at any time and the sequence of initial capitals converges. 
By writing
\beq\label{eq:kn}
A^n(t,\w)-A^{n-1}(t,\w)+k_n=V(t,\w;\phi^n)-V(t,\w;\phi^{n-1}),
\eeq
where $k_n=c_n-c_{n-1}=V^{n}_0(\w_0)-V^{n-1}_0(\w_0)$, we see that the monotonicity of $\{V(t,\w;\phi^n)\}_{n\in\NN}$ is obtained by opportunely choosing a finite $k_n\geq0$ (i.e. by choosing $V^n_0$), which is made possible by the boundedness of $\abs{A^n(t,\w)-A^{n-1}(t,\w)}$, implied by condition \eq{cn-conv}.
However, it is not sufficient to have $k_n<\infty$ for all $n\in\NN$, but we need the convergence of the series $\sum_{n=1}^\infty k_n$.
This is provided again by condition \eq{cn-conv}, because the minimum value of $k_n$ satisfying the positivity of \eq{kn} for all $t\in[0,T]$ is indeed $\max_{t\in[0,T]}(A^n(t,\w)-A^{n-1}(t,\w))^-$.
On the other hand, since both the sequence $\{V(t,\w;\phi^n)\}_{n\geq1}$ for any $t\in[0,T]$ and the sequence $\{V^n_0\}_{n\geq1}$ are regular, i.e. they have limit for $n$ going to infinity, by \eq{An} the sequence $\{A^n(t,\w)\}_{n\geq1}$ is also regular. Finally, since the sequence of initial capitals converges, the equation \eq{An} implies that the sequence of approximations of the quadratic variation of $\w$ converges if and only if $\{G(T,\w;\phi^n)\}_{n\geq1}$ converges.
But $U$ is a plausible set by assumption, thus convergence must hold.
\endproof

\appendix
\section{Comparison of pathwise notions of quadratic variation}

An important distinguish has to be done between \defin{qv1} and the notions of $2$-variation and local 2-variation discussed in the Introduction and on which the theory of extended Riemann-Stieltjes integrals is based (see e.g. \citet[Chapters 1,2]{dudley-norvaisa} and \citet[Section 1]{norvaisa}). 
Let $f$ be any real-valued function on $[0,T]$ and $0<p<\infty$, the \emph{$p$-variation} of $f$ is defined as 
\beq\label{eq:p-var}
v_p(f):=\sup_{\k\in P[0,T]}s_p(f;\k)
\eeq
where $P[0,T]$ is the set of all partitions of $[0,T]$ and
$$s_p(f;\k)=\sum_{i=1}^n\abs{f(t_i)-f(t_{i-1})}^p,\quad\text{for }\k=\{t_i\}_{i=0}^n\in P[0,T].$$
The set of functions with finite $p$-variation is denoted by $\W_p$. We also denote by $\mathrm{vi}(f)$ the variation index of $f$, that is the unique number in $[0,\infty]$ such that
$$\bea{l}v_p(f)<\infty,\quad\mbox{for all }p>\mathrm{vi}(f),\\v_p(f)=\infty,\quad\mbox{for all }p<\mathrm{vi}(f)\ea.$$

For $1<p<\infty$, $f$ has the \emph{local $p$-variation} if the directed function $(s_p(f;\cdot),\mathfrak R)$, where $\mathfrak R:=\{\mathcal R(\k)=\left\{\pi\in P[0,T],\,\k\subset\pi\},\,\k\in P[0,T]\right\}$, converges.
An equivalent characterization of functions with local $p$-variation was introduced by \citet{love-young} and it is given by the Wiener class $\W^*_p$ of functions $f\in\W_p$ such that 
$$\limsup_{\k,\mathfrak R}s_p(f;\k)=\sum_{(0,T]}\abs{\De^-f}^p+\sum_{[0,T)}\abs{\De^+f}^p,$$
where the two sums converge unconditionally. 
We refer to \cite[Appendix A]{norvaisa} for convergence of directed functions and unconditionally convergent sums.
The Wiener class satisfies
$\cup_{1\leq q<p}\W_q\subset\W_p^*\subset \W_p$.

A theory on Stieltjes integrability for functions of bounded $p$-variation was developed by \citet{young36,young38} in the thirties and generalized among others by \cite{dudley-norvaisa99,norvaisa02} around the years 2000. 
According to Young's most well known theorem on Stieltjes integrability, if 
\beq\label{eq:ys}
f\in\W_p,\;g\in\W_q,\quad p^{-1}+q^{-1}>1,\,p,q>0,
\eeq
then the integral $\int_0^Tf\ud g$ exists: in the \emph{Riemann-Stieltjes} sense if $f,g$ have no common discontinuities, in the \emph{refinement Riemann-Stieltjes} sense if $f,g$ have no common discontinuities on the same side, and always in the \emph{Central Young} sense. \cite{dudley-norvaisa99} showed that under condition \eq{ys} also the \emph{refinement Young-Stieltjes} integral always exists.
However, in the applications, we often deal with paths of unbounded 2-variation, like sample paths of the Brownian motion. For example, given a Brownian motion $B$ on a complete probability space $(\O,\F,\PP)$, the pathwise integral $(RS)\!\int_0^Tf\ud B(\cdot,\w)$ is defined in the Riemann-Stieltjes sense, for $\PP$-almost all $\w\in\O$, for any function having bounded $p$-variation for some $p<2$, which does not apply to sample paths of $B$.
In particular, in Mathematical Finance, one necessarily deals with price paths having unbounded 2-variation. In the special case of a market with continuous price paths, \cite{vovk-proba} proved that non-constant price paths must have a variation index equal to 2 and infinite 2-variation in order to rule out \lq arbitrage opportunities of the first kind\rq. 
In the special case where the integrand $f$ is replaced by a smooth function of the integrator $g$, weaker conditions than \eq{ys} on the $p$-variation are sufficient (see \cite{norvaisa02} or the survey in \cite[Chapter 2.4]{norvaisa}) to obtain chain rules and integration-by-parts formulas for extended Riemann-Stieltjes integrals, like the refinement Young-Stieltjes integral, the symmetric Young-Stieltjes integral, the Central Young integral, the Left and Right Young integrals, and others.
However, these conditions are still quite restrictive.

 As a consequence, other notions of quadratic variation were formulated and integration theories for them followed.

\subsubsection*{\follmer's quadratic variation and pathwise calculus}
In 1981, \citet{follmer} derived a pathwise version of the \ito\ formula, conceiving a construction path-by-path of the stochastic integral of a special class of functions. His purely analytic approach does not ask for any probabilistic structure, which may instead come into play only in a later moment by considering stochastic processes that satisfy almost surely, i.e. for almost all paths, a certain condition.
F\"ollmer considers functions on the half line $[0,\infty)$, but we present here his definitions and results adapted to the finite horizon time $[0,T]$.
His notion of quadratic variation is given in terms of weak convergence of measures and is renamed here in his name in order to make the distinguish between the different definitions.
\begin{definition}\label{def:qv-follmer}
Given a dense sequence $\Pi=\{\pi_n\}_{n\geq1}$ of partitions of $[0,T]$, for $n\geq1\; \pi_n=(t_i^n)_{i=0,\ldots,m(n)}$, $0=t_0^n<\ldots<t_{m(n)}^n<\infty$, a \cadlag\ function $x:[0,T]\to\R$ is said to have \emph{F\"ollmer's quadratic variation along} $\Pi$ if the Borel measures 
\beq\label{eq:xin}
\xi_n:=\sum\limits_{i=0}^{m(n)-1}\incrx{x}^2\d_{t_i^n},
\eeq
where $\d_{t_i^n}$ is the Dirac measure centered in $t_i^n$, converge weakly to a finite measure $\xi$ on $[0,T]$ with cumulative function $[x]$ and Lebesgue decomposition 
\beq\label{eq:dec-follmer}
[x](t)=[x]^c(t)+\sum_{0<s\leq t}\De x^2(s),\quad \forall t\in[0,T]
\eeq
where $[x]^c$ is the continuous part.
\end{definition}

\begin{proposition}[Follmer's pathwise \ito\ formula]
  Let $x:[0,T]\to\R$ be a \cadlag\ function having F\"ollmer's quadratic variation along $\Pi$.
  Then, for all $t\in[0,T]$, a function $f\in\C^2(\R)$ satisfies
\begin{align}
  \label{eq:follmer_ito}\nonumber
f(x(t))={}& f(x(0))+\int_0^tf'(x(s-))\ud x(s)+\frac12\int_{(0,t]}f''(x(s-))\ud[x](s) \\ \nonumber
&{}+\sum_{0<s\leq t}\lf f(x(s))-f(x(s-))-f'(x(s-))\De x(s)-\frac12f''(x(s-))\De x(s)^2 \rg \\\nonumber
={}& f(x(0))+\int_0^tf'(x(s-))\ud x(s)+\frac12\int_{(0,t]}f''(x(s))\ud[x]^c(s) \\ 
&{}+\sum_{0<s\leq t}\lf f(x(s))-f(x(s-))-f'(x(s-))\De x(s) \rg,
\end{align}
where the pathwise definition
\beq \label{eq:follmer_int}
\int_0^tf'(x(s-))\ud x(s):=\Limn \sum_{t_i^n\leq t}f'(x(t_i^n))\lf x(t_{i+1}^n\wedge T)-x(t_i^n\wedge T)\rg
\eeq
is well posed by absolute convergence.
\end{proposition}
The integral on the left-hand side of \eq{follmer_int} is referred to as the \emph{F\"ollmer integral} of $f\circ x$ with respect to $x$ along $\Pi$.

In the multi-dimensional case, where $x$ is $\R^d$-valued and $f\in\C^2(\R^d)$, the pathwise \ito\ formula gives
\begin{align} \nonumber
f(x(t))={}& f(x(0))+\int_0^t\nabla f(x(s-))\cdot \ud x(s)+\frac12\int_{(0,t]}\mathrm{tr}\lf \nabla^2f(x(s))\ud[x]^c(s) \rg\\\label{eq:follmer_Dito}
&{}+\sum_{0<s\leq t}\lf f(x(s))-f(x(s-))-\nabla f(x(s-))\cdot\De x(s) \rg
\end{align}
and
$$\int_0^t\nabla f(x(s-))\cdot \ud x(s):=\Limn \sum_{t_i^n\leq t}\nabla f(x(t_i^n))\cdot\incrx{x},$$ 
where $[x]=([x^i,x^j])_{i,j=1,\ldots,d}$ and, for all $t\geq0$,
\begin{align*}
[x^i,x^j](t)={}&\frac12\lf[x^i+x^j](t)-[x^i](t)-[x^j](t)\rg\\
{}={}&[x^i,x^j]^c(t)+\sum_{0<s\leq t}\De x^i(s)\De x^j(s).
\end{align*}
F\"ollmer also pointed out that the class of functions with finite quadratic variation is stable under $\C^1$ transformations and, given $x$ with finite quadratic variation along $\Pi$ and $f\in\C^1(\R^d)$, the composite function $y=f\circ x$ has finite quadratic variation 
$$[y](t)=\int_{(0,t]}\mathrm{tr}\lf \nabla^2f(x(s))^\mathrm{t}\ud[x]^c(s)\rg+\sum_{0<s\leq t}\De y^2(s).$$

\subsubsection*{Norvai\u sa's quadratic variation and chain rules}
Norvai\u sa's notion of quadratic variation was proposed in \cite{norvaisa} in order to weaken the requirement of local 2-variation used to prove chain rules and integration-by-parts formulas for extended Riemann-Stieltjes integrals.
\begin{definition}\label{def:qv-norvaisa}
  Given a dense nested sequence $\l=\{\l_n\}_{n\geq1}$ of partitions of $[0,T]$, \emph{Norvai\u sa's quadratic $\l$-variation} of a regulated function $f:[0,T]\to\R$ is defined, if it exists, as a regulated function $H:[0,T]\to\R$ such that $H(0)=0$ and, for any $0\leq s\leq t\leq T$,
\beq\label{eq:Nqv}
H(t)-H(s)=\Limn s_2(f;\l_n\Cap[s,t]),
\eeq
\beq\label{eq:Njumps}
\De^-H(t)=(\De^-f(t))^2\quad \text{and}\quad \De^+H(t)=(\De^+f(t))^2,
\eeq
where $\l_n\Cap[s,t]:=(\l_n\cap[s,t])\cup\{s\}\cup\{t\}$, $\De^-x(t)=x(t)-x(t-)$, and $\De^+x(t)=x(t+)-x(t)$.
\end{definition}
In reality, Norvai\u sa's original definition is given in terms of an additive upper continuous function defined on the simplex of extended intervals of $[0,T]$, but he showed the equivalence to the definition given here and we chose to report the latter because it allows us to avoid introducing further notations.

Following F\"ollmer's approach in \cite{follmer}, \citet{norvaisa} also proved a chain rule for a function with finite $\l$-quadratic variation, involving a new type of integrals called Left (respectively Right) Cauchy $\l$-integrals.
We report here the formula obtained for the left integral, but a symmetric formula holds for the right integral.
 Given two regulated functions $f,g$ on $[0,T]$ and a dense nested sequence of partitions $\l=\{\l_n\}$, then \emph{the Left Cauchy $\l$-integral} $(LC)\!\int \phi\ud_\l g$ is defined on $[0,T]$ if there exists a regulated function $\Phi$ on $[0,T]$ such that $\Phi(0)=0$ and, for any $0\leq u<v\leq T$, 
$$\bea{c}\Phi(v)-\Phi(u)=\Limn S_{LC}(\phi,g;\l_n\Cap[u,v]),\\
\De^-\Phi(v)=\phi(v-)\De^-g(v),\quad\De^+\Phi(u)=\phi\De^+g(u),\ea$$
where $$S_{LC}(\phi,g;\k):=\sum_{i=0}^{m-1}\phi(t_i)(g(t_{i+1})-g(t_i))\quad\text{for any }\k=\{t_i\}_{i=0}^m.$$
In such a case, denote $(LC)\!\int_u^v\phi\ud_\l g:=\Phi(v)-\Phi(u).$
\begin{proposition}[Proposition 1.4 in \cite{norvaisa}]
  Let $g$ be a regulated function on $[0,T]$ and $\l=\{\l_n\}$ a dense nested sequence of partitions such that $\{t:\,\De^+g(t)\neq0\}\subset\cup_{n\in\NN}\l_n$. The following are equivalent:
  \begin{enumerate}[(i)]
  \item $g$ has Norvai\u sa's $\l$-quadratic variation;
  \item for any $C^1$ function $\phi$, $\phi\circ g$ is Left Cauchy $\l$-integrable on $[0,T]$ and, for any $0\leq u<v\leq T$,
\begin{align}\label{eq:chainrule-LC}
\Phi\circ g(v)-\Phi\circ g(u)={}&(LC)\!\int_u^v(\phi\circ g)\ud_\l g+\frac12\int_u^v(\phi'\circ g)\ud[g]^c_\l\\
&{}+\sum_{t\in[u,v)}\lf\De^-(\Phi\circ g)(t)-(\phi\circ g)(t-)\De^-g(t)\rg \nonumber\\
&{}+\sum_{t\in(u,v]}\lf\De^+(\Phi\circ g)(t)-(\phi\circ g)(t)\De^+g(t)\rg. \nonumber
\end{align}
  \end{enumerate}
\end{proposition}
Note that the change of variable formula \eq{chainrule-LC} gives the F\"ollmer's formula \eq{follmer_ito} when $g$ is right-continuous, and the Left Cauchy $\l$-integral coincides with the F\"ollmer integral along $\l$ defined in \eq{follmer_int}.

\subsubsection*{Vovk's quadratic variation}
\citet{vovk-cadlag} defines a notion of quadratic variation along a sequence of partitions not necessarily dense in $[0,T]$ and uses it to investigate the properties of \lq typical price paths\rq, that are price paths which rule out arbitrage opportunities in his pathwise framework, following a game-theoretic probability approach.
\begin{definition}\label{def:qv-vovk}
  Given a nested sequence $\Pi=\{\pi_n\}_{n\geq1}$ of partitions of $[0,T]$, $\pi_n=(t_i^n)_{i=0,\ldots,m(n)}$ for all $n\in\NN$, a \cadlag\ function $x:[0,T]\to\R$ is said to have \emph{Vovk's quadratic variation along} $\Pi$ if the sequence $\{A^{n,\Pi}\}_{n\in\NN}$ of functions defined by 
$$A^{n,\Pi}(t):=\sum_{i=0}^{m(n)-1}(x(t^n_{i+1}\wedge t)-x(t^n_i\wedge t))^2,\quad t\in[0,T],$$
converges uniformly in time. In this case, the limit is denoted by $A^\Pi$ and called the Vovk's quadratic variation of $x$ along $\Pi$.
\end{definition}
An interesting result in \cite{vovk-cadlag} is that typical paths have the Vovk's quadratic variation along a specific nested sequence $\{\t_n\}_{n\geq1}$ of partitions composed by stopping times and such that, on each realized path $\w$, $\{\t_n(\w)\}_{n\geq1}$ \emph{exhausts} $\w$, i.e. $\{t:\,\De\w(t)\neq0\}\subset\cup_{n\in\NN}\t_n(\w)$ and, for each open interval $(u,v)$ in which $\w$ is not constant, $(u,v)\cap(\cup_{n\in\NN}\t_n(\w))\neq\emptyset$.

The most evident difference between definitions \ref{def:qv1}, \ref{def:qv-follmer}, \ref{def:qv-norvaisa}, \ref{def:qv-vovk} is that the first two of them require the sequence of partitions to be dense, the third one requires the sequence of partitions to be dense and nested, and the last one requires a nested sequence of partitions.
Moreover, Norvai\v sa's definition is given for a regulated, rather than \cadlag, function.

Vovk proved that for a nested sequence $\Pi=\{\pi_n\}_{n\geq1}$ of partitions of $[0,T]$ that exhausts $\w\in D([0,T],\R)$, the following are equivalent:
\begin{enumerate}[(a)]
\item $\w$ has Norvai\v sa's quadratic $\Pi$-variation;
\item $\w$ has Vovk's quadratic variation along $\Pi$;
\item $\w$ has \emph{weak quadratic variation of $\w$ along $\Pi$}, i.e. there exists a \cadlag\ function $V:[0,T]\to\R$ such that
$$V(t)=\Limn\sum_{i=0}^{m(n)-1}(x(t^n_{i+1}\wedge t)-x(t^n_i\wedge t))^2$$
for all points $t\in[0,T]$ of continuity of $V$ and it satisfies \eq{qv-jumps} where $[x]_\Pi$ is replace by $V$.
\end{enumerate}
Moreover, if any of the above condition is satisfied, then $H=A^\Pi=V$.

If, furthermore, $\Pi$ is also dense, than $\w$ has F\"ollmer's quadratic variation along $\Pi$ if and only if it has any of the quadratic variations in (a)-(c), in which case  $H=A^\Pi=V=[\w]$.

In this thesis, we will always consider the quadratic variation of a \cadlag\ path $w$ along a dense nested sequence $\Pi$ of partitions that exhausts $\w$, in which case our \defin{qv1} is equivalent to all the other ones mentioned above. It is sufficient to note that condition (b) implies that $\w$ has finite quadratic variation according to \defin{qv1} and $[\w]=A$, because the properties in  \defin{qv1} imply the ones in \defin{qv-follmer}, which, by Proposition 4 in \cite{vovk-cadlag}, imply condition (b). Therefore, we denote $\bar k(n,t):=\max\{i=0,\ldots,m(n)-1:\,t^n_i\leq t\}$ and note that
\begin{multline*}
A^{n,\Pi}(t)-\sum_{\stackrel{i=0,\ldots,m(n)-1:}{t^n_i\leq t}}(x(t^n_{i+1})-x(t^n_{i}))^2=\\
=(\w(t)-\w(t^n_{\bar k(n,t)}))^2-(\w(t^n_{\bar k(n,t)+1})-\w(t^n_{\bar k(n,t)}))^2\limn 0
\end{multline*}
by right-continuity of $\w$ if $t\in\cup_{n\in\NN}\pi_n$, and by the assumption that $\Pi$ exhausts $\w$ if $t\notin\cup_{n\in\NN}\pi_n$.

\bibliographystyle{plainnat}
\bibliography{thesis}

\subsection*{Acknowledgments}
We particularly thank Rama Cont and Sara Biagini, who provided insight and expertise that greatly assisted the research.
\\
\end{document}